\documentclass[11pt]{article}
\usepackage[english]{babel} 
\usepackage[utf8]{inputenc} 
\usepackage[T1]{fontenc}
\usepackage{lmodern}
\usepackage[top = 3cm ,bottom = 3cm,left= 3cm,right = 3cm]{geometry}
\usepackage{amsfonts}
\usepackage{graphicx}
\usepackage{amsmath}
\usepackage{amssymb}
\usepackage{amsthm}
\usepackage{marvosym}
\usepackage{lipsum}
\usepackage{soul,color}
\usepackage{fancyhdr}
\usepackage{authblk}
\usepackage [
n,
advantage,
operators,
sets,
adversary,
landau,
probability,
notions,
logic,
ff,
mm,
primitives,
events,
complexity,
oracles,
asymptotics,
keys
]{cryptocode}
\createprocedureblock{game}{center,boxed}{}{}{}
\usepackage{lipsum}

\newtheorem{theorem}{Theorem}[section]
\newtheorem{proposition}[theorem]{Proposition}
\newtheorem{definition}[theorem]{Definition}

\usepackage{hyperref}
\usepackage{cleveref}
\crefname{protocol}{Protocol}{Protocols}
\Crefname{protocol}{Protocol}{Protocols}

\newcommand{\conc}{\mathbin\Vert}
\newcommand{\cond}{\, | \,}

\title{A blueprint for constructing 3-pass AKE protocols under commitment-based models}
\author[1,2]{\rm Rodrigo Martín Sánchez-Ledesma \Letter}
\affil[1]{Universidad Complutense de Madrid}
\affil[ ]{\texttt {rodrma01@ucm.es}}
\affil[2]{Indra Sistemas de Comunicaciones Seguras}
\affil[ ]{\texttt {rmsanchezledesma@indra.es}}

\begin{document}

\maketitle
\begin{abstract} 
The commitment-based AKE model of \cite{RMSL:2026:AKE_COM} provides a formal security framework for key exchange protocols that avoid long-term cryptographic material, achieving authentication through a final out-of-band verification of session-derived values. Within this model, secure KA-based and KEM-based protocols were previously constructed via a commitment-based MT compiler, yielding optimized 4-pass protocols. In this work, we show that 3-pass protocols secure under this model exist for both primitives. These protocols are constructed ad hoc, following the core ideas of the commitment-based MT authenticator, and their SK security in the unauthenticated model is proved using the same game-based techniques, achieving bounds of the same form as those of \cite{RMSL:2026:AKE_COM}. The resulting protocols provide one-way authentication in three message exchanges.
\end{abstract}

\section{Introduction}
Secure key exchange over unauthenticated channels is one of the central problems in cryptography. The foundational works of Bellare, Canetti, and Krawczyk [BCK98] and Canetti and Krawczyk [CK01] formalize this problem through a two-layer model: an Authenticated Model (AM), in which messages are delivered faithfully, and an Unauthenticated Model (UM), which permits an active adversary to intercept, modify, and forge messages. Security in the UM is constructed from security in the AM via the notion of a compiler, relying on an initial authenticated exchange of long-term cryptographic material.
The work of \cite{RMSL:2026:AKE_COM} proposes an alternative AKE model that departs from this paradigm in a key respect: instead of an initial authenticated exchange of cryptographic keys, authentication is achieved through a final authenticated verification of session-derived values. This value, called a session Authentication Value (AV), is a deterministic digest of the elements shared between parties throughout the protocol run, verified at the end of the exchange via an out-of-band finalization function $I_f$. The model is intended for settings where long-term cryptographic material is unavailable, impractical, or undesirable, and is particularly suited to ephemeral, session-specific protocols of the kind found in real-time secure communications.
Within this framework, \cite{RMSL:2026:AKE_COM} constructs secure KA-based and KEM-based protocols by applying a commitment-based MT compiler to AM-secure 2-pass protocols. Each application yields a 3-message authenticated exchange, which applied to both messages of the AM protocol produces a 6-pass UM protocol, reducible to 4 passes via message parallelization. At the other extreme, \cite{RMSL:2026:AKE_COM} also establishes that 2-pass protocols cannot be made secure in the UM under this model regardless of the primitive employed: the adversary retains sufficient freedom to mount a successful AV-collision attack. This leaves 3 passes as the only unexplored case, and one that cannot be reached by a direct compiler application — a single use of the compiler authenticates only one of the two AM messages, leaving the other exposed. Closing this gap therefore requires going beyond the compiler and designing an exchange that covers both AM messages within a single commitment-based interaction.
\subsection{Our contribution}
We construct 3-pass protocols secure under the commitment-based AKE model of \cite{RMSL:2026:AKE_COM}, for both KA-based and KEM-based primitives, and prove their SK security in the unauthenticated model.
The protocols are designed ad hoc, following the same core ideas as the commitment-based MT authenticator of \cite{RMSL:2026:AKE_COM}, but without arising as a direct compiler output. The key design principle, shared by both constructions, is to have the initiator commit to a value before the responder has contributed any information, locking the initiator to its intended input prior to observing the responder's contribution. In the KA-based protocol, the initiator commits to its own public key. In the KEM-based protocol, committing to the public key is not viable — doing so would force the responder to wait for the opening before encapsulating, requiring an additional pass — so instead the initiator commits to a fresh random nonce that is subsequently included in the AV computation. In both cases, the opening of the commitment is transmitted in the final message, and both parties compute an AV over all public elements of the exchange, verified through $I_f$.

The security analysis follows the same game-based methodology as \cite{RMSL:2026:AKE_COM}. For each protocol, we prove an emulation result showing that the 3-pass UM protocol $\epsilon$-emulates the corresponding 2-pass AM protocol, with the emulation bound expressed in terms of the hiding, binding, strong-binding, and collision-resistance advantages of the underlying commitment scheme and the output length of the AV random oracle. SK security in the UM then follows from the general emulation theorem of \cite{RMSL:2026:AKE_COM}. The bounds achieved are of the same form as those of the 4-pass protocols.
We note that both protocols provide one-way authentication only: the responder authenticates the identity of the initiator, but not vice versa. This is an inherent structural consequence of the 3-pass design, as a single AV verification suffices for one direction while mutual authentication would require covering both, necessitating additional message rounds. This stands in contrast to the 4-pass protocols of \cite{RMSL:2026:AKE_COM}, which achieve mutual authentication.

\subsection{Organization}
Section 2 recalls the relevant definitions and results from \cite{RMSL:2026:AKE_COM}, including the commitment-based AKE model, the security notions for commitment schemes, and the emulation theorem that underpins SK security in the UM. Section 3 presents the 3-pass KA-based protocol and its security analysis. Section 4 presents the 3-pass KEM-based protocol and the corresponding analysis. Section 5 concludes.

\section{Preliminaries}

\begin{definition}
A key encapsulation mechanism is a triple of algorithms KEM = (KeyGen, Encaps, Decaps). The key generation algorithm KeyGen generates a key pair $(\pk, \sk)$. The encapsulation algorithm Encaps, given a public key value $\pk$, outputs the pair (ct, k), where \textit{ct} is called the encapsulation of a random value \textit{x} that determines the shared key K. The deterministic decapsulation algorithm Decaps, given the secret key $\sk$ and the encapsulation, outputs the same key K by extracting the random value from the encapsulation \textit{ct}.
\end{definition}
\begin{definition}
(Indistinguishability against KEM scheme) We define the IND-\textit{atk} game, $atk \in \{CPA, CCA\}$ as in the figure below and the IND-\textit{atk} advantage of an adversary $\mathcal{A}$ against the above KEM scheme as 
\begin{align*}
    \advantage{\texttt{IND-atk}}{\texttt{KEM}}[(\mathcal{A})] := 2 \cdot\abs{\prob{\texttt{IND-atk$^{\mathcal{A}}$} \Rightarrow 1} - \frac{1}{2}}
\end{align*}

\begin{pchstack}[boxed, center, space=1em]
  {\procedure[linenumbering]{$\indcpa^\adv$ $\pcbox{\indcca^\adv}$}{\phantomsection\label{kemgame}
      (\pk,\sk) \sample \texttt{KGen}  \\
      b \sample \bin  \\
      (ct^*, K_0^*) \sample \texttt{Encaps}(\pk)  \\
      K_1^* \sample \textit{K}  \\
      b' \sample \adv^{\pcbox{Decaps}}(\pk, ct^*, K_b^*) \\
      \pcreturn b = b'
    }}

  {\procedure[linenumbering] {Oracle $Decaps(ct)$}{%
        \text{if ct = ct$^*$} \\
        \text{ return $\bot$} \\
        else \\
        \text{return \texttt{Decaps}(sk,ct)}
   }}
\end{pchstack}
\end{definition}

\begin{definition}
A Commitment Scheme is a triple (Setup, Com, Open) such that:
\begin{itemize}
    \item $CK \sample \texttt{Setup}(\secparam)$ generates the public commitment context. It is often omitted mentioning the public context CK when clear.
    \item for any $ m \in M$, message space, $(c, d) \sample \texttt{Com}$($m$) is the commitment/opening pair for $m$. $c$ = $c(m)$ serves as the commitment value, and $d$ = $d(m)$ as the opening value.
    \item Open($c$, $d$) = $m'\in M \cup \{\bot\}$, where $\bot $ is returned if c is not a valid commitment to any message.
\end{itemize}
We define the hiding and binding games and advantages of an adversary $\mathcal{A}$ and security parameter \textit{n} as follows:
\begin{pchstack}[boxed, center, space=1em]
      {\procedure[linenumbering]{Hiding$^\adv(n)$}{\phantomsection\label{Hiding}
      CK \sample \texttt{Setup}(\secparam) \\
      (x_0, x_1) \sample \adv_1 (\secparam)  \\
      b \sample \bin  \\
      (c(x_b), d(x_b)) \sample \texttt{Com}(x_b)  \\
      b' \sample \adv_2(c(x_b)) \\
      \pcreturn b = b'
    }}
      {\procedure[linenumbering]{Binding$^\adv(n)$}{\phantomsection\label{Binding}
      CK \sample \texttt{Setup}(\secparam) \\
      (c, d, d') \sample \adv (\secparam)  \\
      m = Open(c, d)  \\
      m' = Open(c, d')  \\
      \pcreturn (m \neq m') \wedge (m, m' \neq \bot)
    }}
\end{pchstack}
\begin{center}
    \begin{align*}
        \advantage{\texttt{Hiding}}{\texttt{CS}}[(\mathcal{A}, n)] := 2 \cdot\abs{\prob{\texttt{Hiding$^\adv$} \Rightarrow 1} - \frac{1}{2}} \leq \text{negl(n)} \\
        \advantage{\texttt{Binding}}{\texttt{CS}}[(\mathcal{A}, n)] := \prob{\texttt{Binding$^\adv$} \Rightarrow 1} 
        \leq \text{negl(n)} \\
    \end{align*}
\end{center}
\end{definition}

\begin{definition} \cite{RMSL:2026:AKE_COM}
A Robust Commitment Scheme (RCS) is defined to be a Commitment Scheme such that it verifies the following additional security properties:
\begin{pchstack}[boxed, center, space=1em]
      {\procedure[linenumbering]{Strong-Binding$^\adv(n)$}{\phantomsection\label{Strong-Binding}
      CK \sample \texttt{Setup}(\secparam) \\
      (c, d, d') \sample \adv (\secparam)  \\
      m = Open(c, d)  \\
      m' = Open(c, d')  \\
      \pcreturn (m = m') \wedge (d' \neq d) \wedge (m \neq \bot)
    }}
      {\procedure[linenumbering]{CR$^\adv(n)$}{\phantomsection\label{CR-CS}
      CK \sample \texttt{Setup}(\secparam) \\
      (c, c', d) \sample \adv (\secparam)  \\
      m = Open(c, d)  \\
      m' = Open(c, d')  \\
      \pcreturn (m, m' \neq \bot) \wedge (c' \neq c)
    }}
\end{pchstack}
\begin{center}
    \begin{align*}
        \advantage{\texttt{Strong-Binding}}{\texttt{CS}}[(\mathcal{A}, n)] := \prob{\texttt{Strong-Binding$^\adv$} \Rightarrow 1} 
        \leq \text{negl(n)} \\
        \advantage{\texttt{CR}}{\texttt{CS}}[(\mathcal{A}, n)] := \prob{\texttt{CR$^\adv$} \Rightarrow 1} 
        \leq \text{negl(n)} \\
    \end{align*}
\end{center}
\end{definition}

\begin{definition}
    A cryptographic hash function (CHF) is a function $H: \{0,1\}^{\infty} \rightarrow \{0,1\}^n$, which takes arbitrary length inputs and returns outputs of fixed length $n$. Furthermore, they verify:
    \begin{itemize}
        \item Preimage resistance: Given $y \in \{0,1\}^n$, it is hard to find $x$ such that $H(x) = y$.
        \item Second preimage resistance: Given an input $x$, it is hard to find $x' \neq x$ such that $H(x) = H(x')$.
        \item Collision resistance:  It is hard to find two inputs $x \neq x'$ such that $H(x) = H(x')$.
    \end{itemize}
\end{definition}

Note that, in the above definition, it remains to define what \textit{hard} means. For the sake of this definition, we will consider that to mean that it requires efforts comparable to a brute-force search, i.e. $O(2^n)$. This definition implies that, if $n$ is chosen to be small enough, this wont actually be 'hard' in the computational sense.

Throughout this work, we will model the random functions employed in the commitment model as random oracles. This is done to simplify the proofs when considering the probability of a random value $x$ to have a certain value $H(x) = y$, which under the ROM will be $2^{-n}$. We make the convention that $H(x_1 \conc ... \conc x_n) = \bot$ if any $x_i = \bot$.

Moreover, as per \cite{HHK:2017:MAF}, we make the convention that, in order to keep
record of the queries issued to $H$, we will use a hash list $L_H$ that contains all tuples $(x,H(x))$ of arguments that H was queried on and the respective answers $H(x)$. This way, we can ensure that the number $q$ of queries the adversary makes are different, and that the random oracle is deterministic.

\begin{pchstack}[boxed, center, space=1em]
  {\procedure[linenumbering] {Oracle $H(m)$}{%
        \text{if $\exists r$ s.t $(m, r) \in L_H$} \\
        \text{ return $r$} \\
        r \sample \{0, 1\}^N \\
        L_H = L_H \cup \{(m, r)\} \\
        \text{return $r$}
   }}
\end{pchstack}

Now, we define the central element of the authentication security of our model to give a proper definition for the element that will be used to detect session interference by an attacker:
\begin{definition}\label{AV_G}
AV are constructed as
\begin{align*}
    G(A_1, ..., A_j) := G(enc(A_1) \conc ... \conc enc(A_j))
\end{align*}
where $A_i$ , $i \in \{1, .., j\}$ are elements known to both parties at the end of the protocol session, G is a CHF and $enc$ is an encoding that ensures that collisions cannot happen between inputs. In other words, $enc(x_1) \conc enc(y_1) = enc(x_2) \conc enc(y_2) \iff (x_1, y_1) = (x_2, y_2).$

For the purpose of this work, we will define $enc(A_j)$ as $j \conc len(A_j) \conc A_j$, where $len(x)$ is the length operator. This represents a standard Type-Length-Value (TLV) encoding.
\end{definition}
Within each protocol, the specific values that will conform this AV session value will be established.

We also employ an interesting property that will appear repeatedly throughout the work, defined as the \textit{Combined Collision} advantage of an adversary against a CHF:
\begin{definition}[Combined collision advantage] \cite{RMSL:2026:AKE_COM} Let $l \in \mathbb{N}$ and $\mathcal{Y}$ be an algorithm that returns $y \in Y$ with a certain probability and $\bot$ otherwise. We define the Combined Collision game and the Combined Collision advantage of a PPT adversary $\mathcal{A}$ against a CHF \textbf{G} as 
\begin{align*}
    \advantage{\texttt{combined}}{\texttt{CHF}}[(\mathcal{A}, l, \mathcal{Y})] := \prob{\texttt{combined$^{\mathcal{A}}(l, Y)$} \Rightarrow 1}
\end{align*}

\begin{pchstack}[boxed, center, space=1em]
  {\procedure[linenumbering]{$combined^\adv(l, Y)$}{\phantomsection\label{combined_game}
      T \sample \bin^l \\
      y \sample \adv^{G, \mathcal{Y}}(T) \\
      \pcreturn (G(y) = T) \wedge (T \neq \bot) \wedge (y \in Y \cup \{\bot\})
    }}
\end{pchstack}
where $l$ represents the output length of G and Y represents the domain from which $y$ must be generated.
\end{definition}

The following results provides more details regarding the actual advantage against the above game:
\begin{proposition}\label{combined_adv}\cite{RMSL:2026:AKE_COM}
    For any PPT adversary $\adv$ capable of at most $q$ queries to a CHF random oracle $\mathcal{O}$, it holds that
    $$\advantage{\texttt{combined}}{\texttt{CHF}}[(\mathcal{A}, l, \mathcal{Y})] \leq \frac{q}{2^l} \cdot \delta$$
    where $l$ represents the output length of $\mathcal{O}$, and $\delta := \advantage{\texttt{sample}}{\mathcal{Y}}[(\mathcal{A})]$ represents the probability of algorithm $\mathcal{Y}$ to return $y \in Y$.
\end{proposition}

The following result is a generalization of \cite[Theorem 2]{BdKM:2023:MDK}, which can be applied to the CS-based AKE model as well, which states that the SK-security of a protocol in the UM can be given in terms of the SK-security of another protocol in the AM and the degree of emulation between them:
\begin{theorem}
\label{ThEmulSK}
Let $\pi$ be an $\epsilon$-SK-secure protocol in the AM and let  $\pi'$ a protocol that $\alpha$-emulates $\pi$ in the UM model. Then, $\pi'$ is an $\epsilon'$-SK-secure protocol in the UM model, with $\epsilon' = \epsilon + \alpha$.
\end{theorem}
\begin{proof}
    The exact proof of \cite[Theorem 2]{BdKM:2023:MDK} can be applied to prove this result as well, as it employs no feature specific to the CK-based AKE model that does not translate to the CS-based AKE model nor any feature about $\pi'$ being the result of the application of a compiler.
\end{proof}

\section{An alternative 3-pass KA-based protocol}
This section will present an alternative KA-based protocol constructed from 3-passes, in an attempt to improve the optimized 4-pass KA-based protocol presented in \cite{RMSL:2026:AKE_COM}.

\subsection{A secure 2-pass KA-based protocol in the AM}
For a key exchange of any kind to take place, the minimum number of messages than needs to be exchange is two. In the case of KA-based protocols based on the CS-based AKE model introduced in \cite{RMSL:2026:AKE_COM}, a generic 2-pass protocol would look like this:

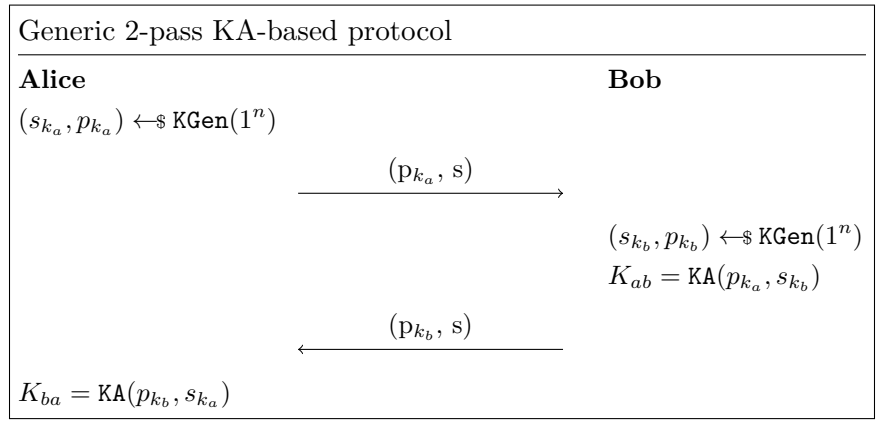
\begin{pchstack}[boxed, center, space=1em]
    \procedure{Generic 2-pass KA-based protocol}{%
     \textbf{Alice} \> \> \textbf{Bob} \\
     (s_{k_a}, p_{k_a}) \sample \texttt{KGen} (\secparam) \> \> \\
     \> \sendmessageright{top=\text{(p$_{k_a}$, s)}} \> \\
     \> \> (s_{k_b}, p_{k_b}) \sample \texttt{KGen} (\secparam) \\ 
     \> \> K_{ab} = \texttt{KA}(p_{k_a}, s_{k_b}) \\ 
     \> \sendmessageleft{top=\text{(p$_{k_b}$, s)}} \> \\ 
     K_{ba} = \texttt{KA}(p_{k_b}, s_{k_a}) \> \> }
\end{pchstack}

In particular, for a DH-based KA scheme, the following proposition gives a bound on the security of a 2-pass construction, in the AM:
\begin{proposition}\cite{CK:2001:AKE}\label{sk_security_2_pass_dh}
    Let $\adv$ be an adversary against the SK-security of the above protocol, which interacts with at most $l$ sessions. Then the 2-pass DH protocol is $\epsilon$-SK-secure in the AM, with 
    $$ \epsilon \leq l \cdot \advantage{DDH}{DH}$$
\end{proposition}

However, it was shown in \cite{RMSL:2026:AKE_COM} that a 2-pass KA-based protocol cannot be secure in the UM, regardless of the KA scheme employed:

\begin{proposition}\cite{RMSL:2026:AKE_COM}
The 2-pass KA-based protocol has SK-security in the UM model, regardless of the KA scheme employed, bounded by the following value:
\begin{align*}
    \advantage{\texttt{Key-Ind}}{\pi_{UM}}[(\adv)] \geq 1 - \left(1 - \frac{1}{2^l}\right)^{\min\{q, |Y|\}}
\end{align*}
where $l$ is the output length of the AV random oracle, $q$ is the maximum number of queries the adversary is able to make against the AV random oracle and $Y: = \{(\pk, K) : (\sk, \pk) \sample \texttt{KGen}(), K \sample \texttt{KA}(\pk_p, \sk)\}$ is the domain of all possible combinations of valid public keys and shared secrets with public key $\pk_p$ of the KA scheme.
\end{proposition}

\subsection{Reducing the number of passes}
In \cite{RMSL:2026:AKE_COM}, a direct application of the CS-based MT compiler to the above 2-pass protocol was employed to derive a secure KA-based protocol in the UM. 

This application yielded, after message optimizations, 4-passes and a two way authentication protocol (i.e. a protocol in which the two messages of the AM-secure 2-pass protocol are authenticated).

We define now an alternative AKE protocol from KA primitives constructed only under 3 passes, which apples directly the same ideas as the CS-based MT authenticator.
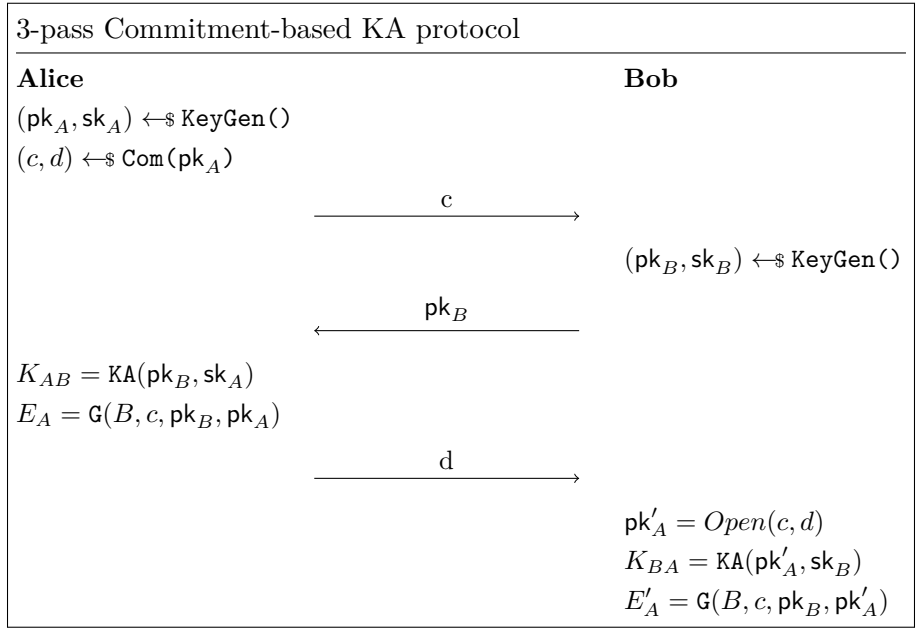
\begin{pchstack}[boxed, center, space=1em]
    \procedure{3-pass Commitment-based KA protocol}{%
     \textbf{Alice} \> \> \textbf{Bob} \\
     (\pk_A, \sk_A) \sample \texttt{KeyGen()} \> \> \\
     (c, d) \sample \texttt{Com($\pk_A$)} \> \> \\
     \> \sendmessageright{top=\text{c}} \> \\
     \> \>(\pk_B, \sk_B) \sample \texttt{KeyGen()} \\ 
     \> \sendmessageleft{top=\text{$\pk_B$}} \> \\ 
     K_{AB} = \texttt{KA}(\pk_B, \sk_A) \> \> \\
     E_A = \texttt{G}(B, c, \pk_B, \pk_A) \> \> \\
     \> \sendmessageright{top=\text{d}} \> \\
     \> \> \pk_A^{\prime} = Open(c, d) \\
     \> \> K_{BA} = \texttt{KA}(\pk_A^{\prime}, \sk_B) \\
     \> \> E_A^{\prime} = \texttt{G}(B, c, \pk_B, \pk_A^{\prime})
     }
\end{pchstack}
The idea behind the above protocol is the following:
\begin{itemize}
    \item Alice generates it own key pair, a commitment \textit{c} of the public key, and sends this commitment to Bob.
    \item Bob, upon reception of such message, just generates its key pair and transmits the public key to Alice. 
    \item When Alice receives Bob's public key, she can proceed to derive the shared secret, and liberates the opening \textit{d} associated with the previous commitment, so Bob can access Alice's public key.
    \item Both parties generate a summary of the exchange, by deriving an AV from the commitment \textit{c}, the public keys of both parties and the identity of the intended recipient $B$, which is added to ensure that an attacker cannot redirect legitimate messages to other uncorrupted not-intended destinations. Once the exchange is finished, the responder will validate its AV against the one obtained by Alice, in a secure and authenticated way, as modeled by $I_f$.
\end{itemize}
The SK-security of the above protocol resides in the security properties of commitment schemes and the elements that conform the AV. 

In order to formalize its SK-security, we need first to prove that the above protocol emulates the generic 2-pass KA-based protocol over unauthenticated networks.

\begin{proposition} \label{3_pass_2_pass_emulation_bound}
The generic 3-pass KA-based protocol, when instantiated with a secure Robust Commitment Scheme CS and a random oracle O, $\epsilon$-emulates the generic 2-pass KA-based protocol over unauthenticated channels with 
\begin{multline*}
\epsilon \leq l \cdot (\frac{3}{2^{n_{av}}} + \frac{q_O}{2^{n_{av}}} \cdot ((1 + \advantage{\texttt{Hiding}}{\texttt{CS}} \cdot \advantage{\texttt{CR}}{\texttt{CS}}) \cdot \advantage{\texttt{Hiding}}{\texttt{CS}} \\ 
+ 2 \cdot \advantage{\texttt{Binding}}{\texttt{CS}} +\advantage{\texttt{Strong-Binding}}{\texttt{CS}}))
\end{multline*} 
where $l = n_p^2 \cdot n_e$, $n_p$ is the number of parties running the protocol, $n_e$ the maximum number of exchanges initiated by each party, $q_O$ is the maximum number of queries that a PPT adversary issues to $O$ and $n_{av}$ the length of the random oracle $O$.
\end{proposition}
\begin{proof}
First, we will define the output "Q validated a KE request from P with public key $\pk$" as the event in which party Q receives the start of an exchange request from party P and public key $\pk$, and successfully follows it through. And the output “P started a KE request to Q with public key $\pk$” as the one in which party P sends an exchange request to party Q employing public key $\pk$.

An AM-adversary can perfectly simulate an UM-adversary unless the event \textit{'In the UM there is the output “Q validated a KE request from P with public key $\pk$” for some parties P and Q, but there was no previous output “P started a KE request to Q with public key $\pk$”, for uncorrupted parties P, Q'} happens. We first bound the probability of this event happening on any particular $(P, Q, \pk)$, and then limit the probability of distinguishing between adversaries by the probability of the above event happening on a particular exchange times the maximum number of exchanges generated.
    
Under the newly defined UM, if the output \textit{"Q validated an exchange request from P with public key $\pk$"} happens, in particular it means that $Q$ has trustfully validated $P$ as the party at the other end of the communication. This means that we can establish $P$ as the sender, and $Q$ as the final receiver. It remains to see that, except with negligible probability, it cannot happen that $P$ intended to contact another party $Q^* \neq Q$ or send another public key $\pk^* \neq \pk$.
    
Therefore, the probability of this event happening on any particular $(P, Q, \pk)$ is reduced to the probability of an UM adversary to generate different elements, or a different receiver, to those intended in a way that the later verification phase $I_f$ between $P$ and $Q$ is successful. We model that in the game G$_{AV}$ and the advantage against the scheme as
\begin{align*}
    \advantage{}{\texttt{CS-KA-3}} := \abs{\prob{\texttt{KA$_{AV}$} \Rightarrow 1}}
\end{align*}
\begin{pchstack}[ boxed , center, space=1em]
  {\procedure[linenumbering]{Game KA$_{AV}$}{\phantomsection\label{KA_CS}
      Q \sample \mathcal{P} \\
      (\sk_P, \pk_P) \sample \texttt{KeyGen()} \\
      (c, d) \sample \texttt{Com($\pk_P$)}  \\
      (Q^*, c^*, d^*) \sample \adv_1^{O}(c, Q) \\
      (\sk_Q, \pk_Q) \sample \texttt{KeyGen()} \\
      \pk_Q^* \sample \adv_2^{O}(c, c^*, d^*, Q, Q^*, \pk_Q) \\
      v_1 = O(Q, c, \pk_Q^*, \pk_P) \\
      d^{**} \sample \adv_3^{O}(c, c^*, d^*, Q, Q^*, \pk_Q, \pk_Q^*, d) \\
      \pk_P^{**} = Open(c^*, d^{**}) \\
      v_2 = O(Q^*, c^*, \pk_Q, \pk_P^{**}) \\
      \pcreturn (v_1 = v_2) \wedge \bigwedge (v_i \neq \bot) \wedge (c^*, Q^*, \pk_Q^*, d^{**}) \neq (c, Q, \pk_Q, d)
    }}
\end{pchstack}
where $\mathcal{P}$ represents the space of valid \textit{parties}, and $O(P, x_1, x_2, x_3)$ represents the AV calculation that will undergo each party and $P$ is the identity of the intended receiver.

It is important to note that, due to Definition \ref{AV_G}, if $(x_1, \cdots, x_n) \neq (x_1^*, \cdots, x_n^*)$, then $TLV(x_1) \conc \cdots \conc TLV(x_n) \neq TLV(x_1^*) \conc \cdots \conc TLV(x_n^*)$. Therefore, for two distinct inputs to generate an output match, a random collision at the oracle $O$ must occur.

We define the intermediate games $KA_0$ and $KA_1$, where intuitively each game represents an additional layer where the attacker behaves like the AM model: in game $G_0$, the attacker does not modify the flow of the last exchange and in game $G_1$ the attacker only modifies the first exchange.

\begin{pchstack}[ boxed , center, space=1em]
  {\procedure[linenumbering]{Game KA$_0$}{\phantomsection\label{G00_KA}
      Q \sample \mathcal{P} \\
      (\sk_P, \pk_P) \sample \texttt{KeyGen()} \\
      (c, d) \sample \texttt{Com($\pk_P$)}  \\
      (Q^*, c^*, d^*) \sample \adv_1^{O}(c, Q) \\
      (\sk_Q, \pk_Q) \sample \texttt{KeyGen()} \\
      \pk_Q^* \sample \adv_2^{O}(c, c^*, d^*, Q, Q^*, \pk_Q) \\
      v_1 = O(Q, c, \pk_Q^*, \pk_P) \\
      \pk_P^{**} = Open(c^*, d) \\
      v_2 = O(Q^*, c^*, \pk_Q, \pk_P^{**}) \\
      b_1 = (v_1 == v_2) \\
      b_2 = \bigwedge (v_i \neq \bot) \\
      b_3 = (c^*, Q^*, \pk_Q^*) \neq (c, Q, \pk_Q) \\
      \pcreturn b_1 \wedge b_2 \wedge b_3
    }}
    
  {\procedure[linenumbering]{Game KA$_1$}{\phantomsection\label{G10_KA}
      Q \sample \mathcal{P} \\
      (\sk_P, \pk_P) \sample \texttt{KeyGen()} \\
      (c, d) \sample \texttt{Com($\pk_P$)}  \\
      (Q^*, c^*, d^*) \sample \adv_1^{O}(c, Q) \\
      (\sk_Q, \pk_Q) \sample \texttt{KeyGen()} \\
      v_1 = O(Q, c, \pk_Q, \pk_P) \\
      \pk_P^{**} = Open(c^*, d) \\
      v_2 = O(Q^*, c^*, \pk_Q, \pk_P^{**}) \\
      b_1 = (v_1 == v_2) \\
      b_2 = \bigwedge (v_i \neq \bot) \\
      b_3 = (c^*, Q^*) \neq (c, Q) \\
      \pcreturn b_1 \wedge b_2 \wedge b_3
    }}
\end{pchstack}

Therefore, the value $\abs{\prob{\texttt{KA$_{AV}$} \Rightarrow 1} - \prob{\texttt{KA$_0$} \Rightarrow 1}}$ represents the advantage of the attacker in generating a collision between the oracle output from both parties when being able to modify the opening value received by the responder. \vspace{1mm}

\textbf{Step 1) Analysis of $\abs{\prob{\texttt{KA$_{AV}$} \Rightarrow 1} - \prob{\texttt{KA$_0$} \Rightarrow 1}}$:} In order to bound this difference, let us define $\mathcal{F}$ as the event that, in Game KA$_{AV}$, an attacker has the ability to exploit the additional step $\adv_3$ to win the game. Formally, this means the event the attacker is able to generate a value $d^{**}$ in step $\adv_3$ such that $Open(c^*, d^{**}) = \pk_P^{**}$, satisfying $v_1 = v_2 \wedge (c^*, Q^*, \pk_Q, d^{**}) \neq (c, Q, \pk_Q^*, d)$. 

Note that, by definition of the games under consideration, since the event $\mathcal{F}$ is the only difference between the games, the application of the difference lemma \cite{cryptoeprint:2004/332} yields
$$\abs{\prob{\texttt{KA$_{AV}$} \Rightarrow 1} - \prob{\texttt{KA$_0$} \Rightarrow 1}} \leq \prob{\mathcal{F} \Rightarrow 1}$$ 

Now, it is required to bound the probability of the event $\mathcal{F}$. To do so, two events are considered: \texttt{F$_1$}, in which $(c^*, Q^*, \pk_Q^*) \neq (c, Q, \pk_Q)$, and \texttt{F$_2$}, in which $(c^*, Q^*, \pk_Q^*) = (c, Q, \pk_Q)$. \vspace{1mm}

\textbf{Event \texttt{F$_1$}:} Under \texttt{F$_1$}, $(c^*, Q^*, \pk_Q, \pk_P^{**}) \neq (c, Q, \pk_Q^*, \pk_P)$ and therefore the only possible way to achieve $v_1 = v_2$ is via an oracle collision. Therefore, an adversary needs to find $d^{**}$ such that $Open(c^*, d^{**}) = \pk_P^{**}$ and $O(Q^*, c^*, \pk_Q, \pk_P^{**}) = O(Q, c, \pk_Q^*, \pk_P)$.

For this analysis, two sub-cases are defined: given $\pk_P^* = Open(c^*, d^*)$, define \texttt{F$_{1, 1}$} as the event that $O(Q^*, c^*, \pk_Q, \pk_P^*) = O(Q, c, \pk_Q^*, \pk_P)$, and \texttt{F$_{1, 2}$} as the event that $O(Q^*, c^*, \pk_Q, \pk_P^*) \neq O(Q, c, \pk_Q^*, \pk_P)$. \vspace{1mm}

\textbf{Sub-event \texttt{F$_{1, 1}$}:} If event \texttt{F$_{1, 1}$} happens, the adversary has already won the game, as a random collision has happened. Therefore, $\prob{\texttt{F} \cond \texttt{F$_{1, 1}$}, \texttt{F$_{1}$}} = 1$ and $\prob{\texttt{F$_{1, 1}$} \cond \texttt{F$_{1}$}} = \frac{1}{2^{n_{av}}}$. \vspace{1mm}

\textbf{Sub-event \texttt{F$_{1, 2}$}:} If \texttt{F$_{1, 2}$} holds, the adversary is required to find $d^{**}$ such that $Open(c^*, d^{**}) = \pk_P^{**} \neq \pk_P^*$ satisfying $O(Q^*, c^*,  \pk_Q, \pk_P^{**}) = O(Q, c, \pk_Q^*, \pk_P)$. Note that the requirement $\pk_P^{**} \neq \pk_P^*$ is precisely because $\texttt{F$_{1, 2}$}$ ensures that, if $\pk_P^{**} = \pk_P^*$, then $O(Q^*, c^*,  \pk_Q, \pk_P^{**}) = O(Q^*, c^*, \pk_Q, \pk_P^*) \neq O(Q, c, \pk_Q^*, \pk_P)$. 

We claim that the probability $\prob{\texttt{F} \cond \texttt{F$_{1, 2}$}, \texttt{F$_{1}$}}$ is bounded by the \textit{Combined Collision} advantage: first note that the probability of an adversary to find a $d^{**}$ such that $Open(c^*, d^{**}) = \pk_P^{**} \neq \pk_P^*$ is modeled exactly by the following: 
\begin{pchstack}[boxed, center]
  {\procedure[linenumbering]{Game G$_{Binding^*}$}{\phantomsection\label{Gbinding'}
      (c, d) \sample \adv_0() \\
      d' \sample \adv (c, d)  \\
      m = Open(c, d)  \\
      m' = Open(c, d')  \\
      \pcreturn (m \neq m') \wedge (m, m' \neq \bot)
    }}
\end{pchstack}
as the pair $(c, d)$ was also generated by the adversary at a previous stage. Note that
$$\prob{\texttt{G$_{Binding^*}$} \Rightarrow 1} \leq \prob{\texttt{G$_{Binding}$} \Rightarrow 1}$$ 
(i.e. the game $G_{Binding^*}$ is bounded by the game $G_{Binding}$), as in the later, the adversary is the one to draw the tuple $(c, d, d')$ at the same time, as opposed to receiving $(c, d)$ and just generating $d'$. Note that, in this case, the games would actually be equivalent, as the two generation phases correspond to the adversary and not additional information is needed between phases.

Now, we define an adversary $\mathcal{Y}_1$ that returns $\pk_P^{**}$ if it can find $d^{**}$ such that $Open(c^*, d^{**}) = \pk_P^{**} \neq \pk_P^*$, which by the above analysis, has probability of success bounded by $\advantage{\texttt{Binding}}{\texttt{CS}}$. Then, $\prob{\texttt{F} \cond \texttt{F$_{1, 2}$}, \texttt{F$_{1}$}}$ is exactly $\advantage{\texttt{combined}}{\texttt{CHF}}[(\mathcal{A}, n_{av}, \mathcal{Y}_1)]$ and thus by \ref{combined_adv} $$\prob{\texttt{F} \cond \texttt{F$_{1, 2}$}, \texttt{F$_{1}$}} \leq q_O \frac{\advantage{\texttt{Binding}}{\texttt{CS}}}{2^{n_{av}}}$$

\textbf{Event \texttt{F$_2$}:} Under \texttt{F$_2$}, $(c^*, Q^*, \pk_Q^*) = (c, Q, \pk_Q)$ and therefore there are only two possible ways to win the game \texttt{F}. The first one, referred to as \texttt{F$_{2, 1}$}, is to generate a collision on the inputs to the oracle function. This means, to find $d^{**} \neq d$ such that $Open(c^*, d^{**}) = Open(c, d^{**}) := \pk_P^{**} = \pk_P$. Note that the requirement $d^{**} \neq d$ is to ensure that a collision on the modified elements does not occur. \vspace{1mm} 

\textbf{Sub-event \texttt{F$_{2, 1}$}:} The probability of such a collision is modeled exactly by the following game
\begin{pchstack}[boxed, center]
  {\procedure[linenumbering]{Game G$_{Strong-Binding^*}$}{\phantomsection\label{GStr_binding'}
  	 m \sample \mathcal{M} \\
      (c, d) \sample \texttt{Com}(m) \\
      d^{**} \sample \adv (c, d)  \\
      m^{**} = Open(c, d^{**})  \\
      \pcreturn (m = m^{**}) \wedge (d^{**} \neq d) \wedge (m \neq \bot)
    }}
\end{pchstack}
and $$\prob{\texttt{G$_{Strong-Binding^*}$} \Rightarrow 1} \leq \prob{\texttt{G$_{Strong-Binding}$} \Rightarrow 1}$$ 
(i.e. the game $G_{Strong-Binding^*}$ is trivially harder that the game $G_{Strong-Binding}$), as in the later, the adversary is the one to draw the tuple $(c, d, d')$, as opposed to receiving $(c, d)$ and just generating $d'$.

Therefore, $$\prob{\texttt{F$_{2, 1}$} \cond \texttt{F$_2$}} = \prob{\texttt{G$_{Strong-Binding^*}$} \Rightarrow 1} \leq \advantage{\texttt{Strong-Binding}}{\texttt{CS}}$$

The second one, referred to as \texttt{F$_{2, 2}$}, is to generate a collision on the outputs, provided that the inputs are not equal. This means, to find $d^{**} \neq d$ such that $Open(c^*, d^{**}) = \pk_P^{**} \neq \pk_P$ and $O(Q^*, c^*,  \pk_Q, \pk_P^{**}) = O(Q, c, \pk_Q^*, \pk_P)$. \vspace{1mm}

\textbf{Sub-event \texttt{F$_{2, 2}$}:} In a similar manner to above, we claim that the probability $\prob{ \texttt{F$_{2, 2}$} \cond \texttt{F$_{2}$}}$ is bounded by the \textit{Combined Collision} advantage: first note that the probability of an adversary to find a $d^{**} \neq d$ such that $Open(c, d^{**}) = \pk_P^{**} \neq \pk_P$ is modeled exactly by the Game G$_{Binding^*}$.

Now, we define an adversary $\mathcal{Y}_2$ that returns $\pk_P^{**}$ if it can find $d^{**} \neq d$ with the above conditions. Then, $\prob{\texttt{F$_{2, 2}$} \cond \texttt{F$_{2}$}}$ is exactly $\advantage{\texttt{combined}}{\texttt{CHF}}[(\mathcal{A}, n_{av}, \mathcal{Y}_2)]$ and thus by \ref{combined_adv} 
$$\prob{ \texttt{F$_{2, 2}$} \cond \texttt{F$_{2}$}} \leq \prob{\texttt{G$_{Binding^*}$} \Rightarrow 1} \cdot \frac{q_O}{2^{n_{av}}} \leq \advantage{\texttt{Binding}}{\texttt{CS}} \cdot \frac{q_O}{2^{n_{av}}}$$

Together, we have 
\begin{align*}
& \prob{\texttt{F} \Rightarrow 1} = \prob{\texttt{F} \Rightarrow 1 \cond \texttt{F$_1$}} \cdot \prob{\texttt{F$_1$}} + \prob{\texttt{F} \Rightarrow 1 \cond \texttt{F$_2$}} \cdot \prob{\texttt{F$_2$}} \leq \\
& \prob{\texttt{F} \Rightarrow 1 \cond \texttt{F$_{1, 1}$}, \texttt{F$_1$}} \cdot \prob{\texttt{F$_{1, 1}$} \cond \texttt{F$_1$}} + \prob{\texttt{F} \Rightarrow 1 \cond \texttt{F$_{1, 2}$}, \texttt{F$_1$}} \cdot \prob{\texttt{F$_{1, 2}$} \cond \texttt{F$_1$}} + \prob{\texttt{F} \Rightarrow 1 \cond \texttt{F$_2$}} \leq \\
& \frac{1}{2^{n_{av}}}\cdot (1 + q_O \cdot \advantage{\texttt{Binding}}{\texttt{CS}}) + \prob{\texttt{F$_{2, 1}$} \Rightarrow 1 \cond \texttt{F$_2$}} + \prob{\texttt{F$_{2, 2}$} \Rightarrow 1 \cond \texttt{F$_2$}} \leq \\
& \frac{1}{2^{n_{av}}}\cdot (1 + 2 \cdot q_O \cdot \advantage{\texttt{Binding}}{\texttt{CS}}) +\advantage{\texttt{Strong-Binding}}{\texttt{CS}}
\end{align*}

The value $\abs{\prob{\texttt{KA$_0$} \Rightarrow 1} - \prob{\texttt{KA$_1$} \Rightarrow 1}}$ represents the advantage of the attacker in generating a collision between the oracle output from both parties when being able to modify the random challenge value sent by the responder. \vspace{1mm}

\textbf{Step 2) Analysis of $\abs{\prob{\texttt{KA$_{0}$} \Rightarrow 1} - \prob{\texttt{KA$_1$} \Rightarrow 1}}$:} In order to bound this difference, let us define $\mathcal{H}$ as the event that in game \texttt{G$_0$} an adversary has the ability to exploit the second exchange to win the game, i.e. to generate $\pk_Q^*$ in step $\adv_2$ such that $v_1 = v_2$ and $(c^*, Q^*, \pk_Q^*) \neq (c, Q, \pk_Q)$. Note that, by definition of the games under consideration, the application of the difference lemma \cite{cryptoeprint:2004/332} yields
$$\abs{\prob{\texttt{KA$_{0}$} \Rightarrow 1} - \prob{\texttt{KA$_1$} \Rightarrow 1}} \leq \prob{\mathcal{H} \Rightarrow 1}$$
Now, it is required to bound the probability of the above event. There are exactly two ways: via a random generation of $\pk_Q^*$ and collision between outputs (referred to as \texttt{H$_1$}), and by trying to extract information and generating a targeted value for $\pk_Q^*$ (i.e. not trying random collision), referred to as \texttt{H$_2$}. \vspace{1mm}

\textbf{Event \texttt{H$_1$}:} The advantage in winning the game from the first option is bounded by the probability of a random collision, which is $\frac{1}{2^{n_{av}}}$. Therefore, we will focus on the probabilities for an adversary in making an elaborated guess. \vspace{1mm}

\textbf{Event \texttt{H$_2$}:} The analysis is done in terms of the event that the adversary knows the value of $\pk_P$, referred to as $D$. First, we will analyze the probability of $D$ itself. Since $\pk_P = Open(c, d)$, the pair $(c, d)$ represents a valid opening of $\pk_P$. Therefore, this probability is bounded by the advantage in retrieving the committed input from a valid commitment. This is modeled by the following game:
\begin{pchstack}[boxed, center]
  {\procedure[linenumbering]{Game G$_{Hiding^*}$}{\phantomsection\label{G_Hiding'}
      m \sample \mathcal{M} \\
      (c, d) \sample Com(m) \\
      m' \sample \adv(c) \\
      \pcreturn (m = m')
    }}
\end{pchstack}

And, in turn, this advantage is bounded by the ability to distinguish between two committed inputs from a commitment value, in a traditional $IND \implies OW$ argument. The latter represents the Hiding advantage of a commitment scheme, and therefore it holds that $\prob{D} \leq \advantage{\texttt{Hiding}}{\texttt{CS}}$. \vspace{1mm}

\textbf{Sub-event \texttt{$\neg D$}:} In case $D$ does not hold, the adversary does not possess all the information required to make an informed guess: the value $v_1$ is dependent on both the choice of $\pk_Q^*$ of the adversary and $\pk_P$, which is defined but not known to the adversary. Therefore, $\prob{\texttt{H$_2$} \cond \neg D} = 0$, as the only possible way when not all information is known is random collision. \vspace{1mm}


\textbf{Sub-event \texttt{$D$}:} In case $D$ does hold, we analyze the success probability in terms of the event $C$ that $c^* = c$. If $C$ is also true, then $Open(c^*, d) = \pk_P^{**} = \pk_P = Open(c, d)$. Therefore, since $\pk_P$ is known, $\pk_P^{**}$ is also known. This provides the adversary with the knowledge of all elements involved in the generation of both $v_1$ and $v_2$. 

Thus the success probability $\prob{\texttt{H$_2$} \cond C,D}$ resides in the adversaries ability to generate a random element $\pk_Q^* \neq \pk_Q$ such that $O(Q^*,c, \pk_Q, \pk_P) = O(Q, c, \pk_Q^*, \pk_P)$. Since the attacker cannot win the game via input collision, as it must hold that $(Q^*, c^*, \pk_Q^*) \neq (Q, c, \pk_Q)$, the only viable option is to search for a preimage-collision over public keys $\pk_Q^*$, which has probability at most $\frac{q_O}{2^{n_{av}}}$, with $q_O$ being the maximum number of queries made to the oracle $O$.

In case $D$ holds but $C$ does not, further information cannot be claimed about $\pk_P^{**}$. For any adversary to be able to win the game without random collision, three conditions must happen simultaneously: first, $\pk_P^{**} \neq \bot$. Second, to be able to retrieve $\pk_P^{**}$ from $c^*$. Lastly, the ability to force a collision via generating $\pk_Q^*$.

For the analysis of first condition, since $\pk_P^{**} = Open(c^*, d)$ and $c^* \neq c$, it is required to bound the probability that $\pk_P^{**} \neq \bot$. It can be bounded by the advantage of the CR property of commitment schemes. Note, however, that this bound is very conservative, since the CR-CS game allows the adversary to find $(c, c', d)$, while in here they are all established.

The second condition requires to bound the probability of retrieving its value. Since $\pk_P^{**} = Open(c^*, d) \neq \bot$, the pair $(c^*, d)$ form a valid commitment pair of the value $\pk_P^{**}$. Therefore, this probability is bounded by the advantage in retrieving the committed input from a valid commitment. As above, this can be bounded by the Hiding property of the CS.

The third condition requires to bound the probability of being able search for a preimage-collision over public keys $\pk_Q^*$. It has probability at most $\frac{q_O}{2^{n_{av}}}$, with $q_O$ being the maximum number of queries made to the oracle $O$.

Together, this means that $$\prob{\texttt{H$_2$} \cond \neg C,D} \leq \advantage{\texttt{Hiding}}{\texttt{CS}} \cdot \advantage{\texttt{CR}}{\texttt{CS}} \cdot \frac{q_O}{2^{n_{av}}}$$
and therefore
\begin{align*}
& \prob{\texttt{H} \Rightarrow 1} \leq \prob{\texttt{H$_1$}} + \prob{\texttt{H$_2$}} \leq \\
& \frac{1}{2^{n_{av}}} + \prob{\texttt{H$_2$} \cond \neg D} \cdot \prob{\neg D} + \prob{\texttt{H$_2$} \cond D} \cdot \prob{D} = \\
& \frac{1}{2^{n_{av}}} + (\prob{\texttt{H$_2$} \cond C, D} \cdot \prob{C \cond D} + \prob{\texttt{H$_2$} \cond \neg C, D} \cdot \prob{\neg C \cond D}) \cdot \prob{D} \leq \\
& \frac{1}{2^{n_{av}}} + (1 + \advantage{\texttt{Hiding}}{\texttt{CS}} \cdot \advantage{\texttt{CR}}{\texttt{CS}}) \cdot \frac{q_O}{2^{n_{av}}} \cdot \advantage{\texttt{Hiding}}{\texttt{CS}}
\end{align*}

\textbf{Step 3) Analysis of $\prob{\texttt{KA$_1$} \Rightarrow 1}$:} Lastly, the advantage of the game $G_{1}$ is exactly the probability of a random collision between inputs when generating a different commitment value $c^*$, or when specifying a different receiver $Q^*$ than the one $Q$ originally intended.

It is straightforward to see that, regardless of the information that can be extracted from $c$ or generated from it, a random collision on the oracle is still required, since $\pk_Q$ has not yet been generated and $(c^*, Q^*) \neq (c, Q)$. The probability of this collision is exactly $\frac{1}{2^{n_{av}}}$.

Together, we have the following:
\begin{align*}
    & \abs{\prob{\texttt{KA$_{AV}$} \Rightarrow 1}} \leq \\ 
    & \abs{\prob{\texttt{KA$_{AV}$} \Rightarrow 1} - \prob{\texttt{KA$_0$} \Rightarrow 1}} + \abs{\prob{\texttt{KA$_{0}$} \Rightarrow 1} - \prob{\texttt{KA$_1$} \Rightarrow 1}} + \abs{\prob{\texttt{KA$_1$} \Rightarrow 1}} \leq \\
    & \frac{3}{2^{n_{av}}} + \frac{q_O}{2^{n_{av}}} \cdot ((1 + \advantage{\texttt{Hiding}}{\texttt{CS}} \cdot \advantage{\texttt{CR}}{\texttt{CS}}) \cdot \advantage{\texttt{Hiding}}{\texttt{CS}} \\
    & + 2 \cdot \advantage{\texttt{Binding}}{\texttt{CS}} +\advantage{\texttt{Strong-Binding}}{\texttt{CS}})
\end{align*}
    Then, considering all possible messages between parties, which equals the maximum number of messages sent by each party times the square of the number of parties, we arrive to the desired bound.
\end{proof}

It is important to note that this protocol only provides one-way authentication, as only one AV calculation is employed.

Now, the SK security of the 3-pass KA-based protocol depends upon the practical KA scheme employed:
\begin{proposition}\label{gen_3_pass_KA_sk_security}
    Let the 2-pass KA-based protocol be a $\epsilon$-SK-secure protocol in the AM. Then, the 3-pass KA-based protocol is $(\epsilon + 2\cdot \alpha)$-SK secure in the UM,
    where $\alpha$ is the emulation bound proven in Proposition \ref{3_pass_2_pass_emulation_bound}.
\end{proposition}
\begin{proof}
    The result is a direct application of \ref{ThEmulSK} and \ref{3_pass_2_pass_emulation_bound}.
\end{proof}

\section{An alternative 3-pass KEM-based protocol}
This section will present an alternative KEM-based protocol constructed from 3-passes, in an attempt to improve the optimized 4-pass KEM-based protocol presented in \cite{RMSL:2026:AKE_COM}.

\subsection{A secure 2-pass KEM-based protocol in the AM}
In the case of KEM-based protocols based on the CS-based AKE model introduced in \cite{RMSL:2026:AKE_COM}, a generic 2-pass protocol would look like this:

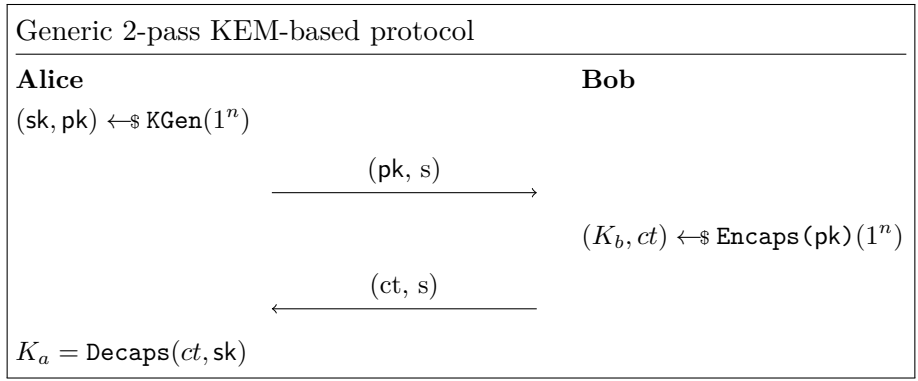
\begin{pchstack}[boxed, center, space=1em]
    \procedure{Generic 2-pass KEM-based protocol}{%
     \textbf{Alice} \> \> \textbf{Bob} \\
     (\sk, \pk) \sample \texttt{KGen} (\secparam) \> \> \\
     \> \sendmessageright{top=\text{($\pk$, s)}} \> \\
     \> \> (K_b, ct) \sample \texttt{Encaps($\pk$)} (\secparam) \\ 
     \> \sendmessageleft{top=\text{(ct, s)}} \> \\ 
     K_{a} = \texttt{Decaps}(ct, \sk) \> \> }
\end{pchstack}

The following proposition gives a bound on the security of a 2-pass construction from a generic KEM scheme, in the AM:
\begin{proposition}\cite{BdKM:2023:MDK}\label{sk_security_2_pass_kem}
    Let $\adv$ be an adversary against the SK-security of the above protocol, which interacts with at most $q$ sessions for each pair of $n_P$ parties. Then, the 2-pass KEM-based protocol is $\epsilon$-SK-secure in the AM, with 
    $$ \epsilon \leq q\cdot n_P^2 \cdot \advantage{CPA}{KEM}$$
\end{proposition}

However, it was shown in \cite{RMSL:2026:AKE_COM} that the security of a 2-pass KEM-based protocol has a lower bound not dependent upon the:

\begin{proposition}\cite{RMSL:2026:AKE_COM}
The 2-pass KA-based protocol has SK-security in the UM model, regardless of the KA scheme employed, bounded by the following value:
\begin{align*}
    \advantage{\texttt{Key-Ind}}{\pi_{UM}}[(\adv)] \geq 1 - \left(1 - \frac{1}{2^l}\right)^{\min\{q, |Y|\}}
\end{align*}
where $l$ is the output length of the AV random oracle, $q$ is the maximum number of queries the adversary is able to make against the AV random oracle and $Y: = \{(\pk, K) : (\sk, \pk) \sample \texttt{KGen}(), K \sample \texttt{KA}(\pk_p, \sk)\}$ is the domain of all possible combinations of valid public keys and shared secrets with public key $\pk_p$ of the KA scheme.
\end{proposition}

\subsection{Reducing the number of passes}
In \cite{RMSL:2026:AKE_COM}, a direct application of the CS-based MT compiler to the above 2-pass protocol was employed to derive a secure KEM-based protocol in the UM. 

This application yielded, after message optimizations, 4-passes and a two way authentication protocol (i.e. a protocol in which the two messages of the AM-secure 2-pass protocol are authenticated).

We define now an alternative AKE protocol from KEM schemes constructed only under 3 passes, which apples directly the same ideas as the CS-based MT authenticator.
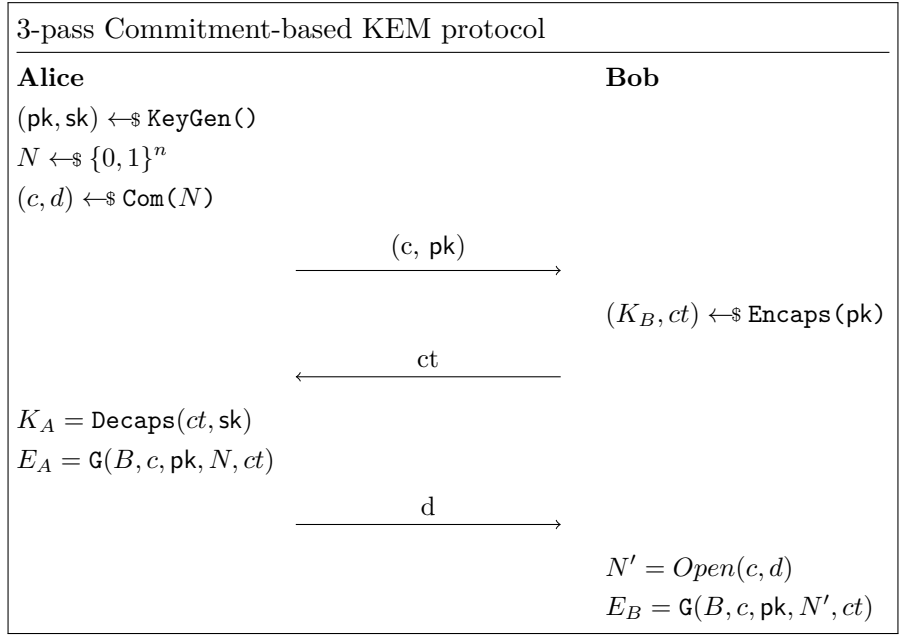
\begin{pchstack}[boxed, center, space=1em]
    \procedure{3-pass Commitment-based KEM protocol}{%
     \textbf{Alice} \> \> \textbf{Bob} \\
     (\pk, \sk) \sample \texttt{KeyGen()} \> \> \\
     N \sample \{0,1\}^n \> \> \\
     (c, d) \sample \texttt{Com($N$)} \> \> \\
     \> \sendmessageright{top=\text{(c, $\pk$)}} \> \\
     \> \>(K_B, ct) \sample \texttt{Encaps($\pk$)} \\ 
     \> \sendmessageleft{top=\text{ct}} \> \\ 
     K_{A} = \texttt{Decaps}(ct, \sk) \> \> \\
     E_A = \texttt{G}(B, c, \pk, N, ct) \> \> \\
     \> \sendmessageright{top=\text{d}} \> \\
     \> \> N^{\prime} = Open(c, d) \\
     \> \> E_B = \texttt{G}(B, c, \pk, N^{\prime}, ct)
     }
\end{pchstack}
The idea behind the above protocol is the following:
\begin{itemize}
    \item Alice generates a key pair. Since a commitment of the public key would completely block the capacity of ending the protocol in three passes, Alice generates commitment \textit{c} of a random element $N$, and sends this commitment to Bob, along with the public key. It will be shown that a commitment of a random element that is involved on the AV calculation will be sufficient.
    \item Bob, upon reception of such message, just generates an encapsulation using the public ket and transmits it to Alice. 
    \item When Alice receives the encapsulation, she can proceed to derive the shared secret via a decapsulation procedure, and liberates the opening \textit{d} associated with the previous commitment, so Bob can access Alice's random value.
    \item Both parties generate a summary of the exchange, by deriving an AV from the commitment \textit{c}, the public elements of both parties, the identity of the intended recipient $B$, which is added to ensure that an attacker cannot redirect legitimate messages to other uncorrupted not-intended destinations, and Alice's random element. Once the exchange is finished, the responder will validate its AV against the one obtained by Alice, in a secure and authenticated way, as modeled by $I_f$.
\end{itemize}
In order to formalize its SK-security, we follow the exact same ideas as Section 3. First, we prove that the above protocol emulates the generic 2-pass KEM-based protocol over unauthenticated networks.

\begin{proposition} \label{3_pass_2_pass_emulation_bound_kem}
The generic 3-pass KEM-based protocol, when instantiated with a secure Robust Commitment Scheme CS and a random oracle O, $\epsilon$-emulates the generic 2-pass KEM-based protocol over unauthenticated channels with 
\begin{multline*}
\epsilon \leq l \cdot (\frac{3}{2^{n_{av}}} + \frac{q_O}{2^{n_{av}}} \cdot ((1 + \advantage{\texttt{Hiding}}{\texttt{CS}} \cdot \advantage{\texttt{CR}}{\texttt{CS}}) \cdot \advantage{\texttt{Hiding}}{\texttt{CS}} \\ 
+ 2 \cdot \advantage{\texttt{Binding}}{\texttt{CS}} +\advantage{\texttt{Strong-Binding}}{\texttt{CS}}))
\end{multline*} 
where $l = n_p^2 \cdot n_e$, $n_p$ is the number of parties running the protocol, $n_e$ the maximum number of exchanges initiated by each party, $q_O$ is the maximum number of queries that a PPT adversary issues to $O$ and $n_{av}$ the length of the random oracle $O$.
\end{proposition}
\begin{proof}
First, we will define the output "Q validated a KE request from P with public key $\pk$" as the event in which party Q receives the start of an exchange request from party P and public key $\pk$, and successfully follows it through. And the output “P started a KE request to Q with public key $\pk$” as the one in which party P sends an exchange request to party Q employing public key $\pk$.

An AM-adversary can perfectly simulate an UM-adversary unless the event \textit{'In the UM there is the output “Q validated a KE request from P with public key $\pk$” for some parties P and Q, but there was no previous output “P started a KE request to Q with public key $\pk$”, for uncorrupted parties P, Q'} happens. We first bound the probability of this event happening on any particular $(P, Q, \pk)$, and then limit the probability of distinguishing between adversaries by the probability of the above event happening on a particular exchange times the maximum number of exchanges generated.
    
Under the newly defined UM, if the output \textit{"Q validated an exchange request from P with public key $\pk$"} happens, in particular it means that $Q$ has trustfully validated $P$ as the party at the other end of the communication. This means that we can establish $P$ as the sender, and $Q$ as the final receiver. It remains to see that, except with negligible probability, it cannot happen that $P$ intended to contact another party $Q^* \neq Q$ or send another public key $\pk^* \neq \pk$.
    
Therefore, the probability of this event happening on any particular $(P, Q, \pk)$ is reduced to the probability of an UM adversary to generate different elements, or a different receiver, to those intended in a way that the later verification phase $I_f$ between $P$ and $Q$ is successful. We model that in the game G-KEM$_{AV}$ and the advantage against the scheme as
\begin{align*}
    \advantage{}{\texttt{CS-KEM-3}} := \abs{\prob{\texttt{KEM$_{AV}$} \Rightarrow 1}}
\end{align*}
\begin{pchstack}[ boxed , center, space=1em]
  {\procedure[linenumbering]{Game KEM$_{AV}$}{\phantomsection\label{KEM_CS}
      Q \sample \mathcal{P} \\
      (\sk, \pk) \sample \texttt{KeyGen()} \\
      N \sample \{0,1\}^n \\
      (c, d) \sample \texttt{Com(N)}  \\
      (Q^*, c^*, d^*, \pk^*) \sample \adv_1^{O}(c, \pk, Q) \\
      ct \sample \texttt{Encaps($\pk^*$)} \\
      ct^* \sample \adv_2^{O}(c, c^*, d^*, \pk, \pk^*, Q, Q^*, ct) \\
      v_1 = O(Q, c, \pk, N, ct^*) \\
      d^{**} \sample \adv_3^{O}(c, c^*, d^*, \pk, \pk^*, Q, Q^*, ct, ct^*, d) \\
      N^{**} = Open(c^*, d^{**}) \\
      v_2 = O(Q^*, c^*, \pk^*, N^{**}, ct) \\
      \pcreturn (v_1 = v_2) \wedge \bigwedge (v_i \neq \bot) \wedge (c^*, Q^*, \pk^*, ct^*, d^{**}) \neq (c, Q, \pk, ct, d)
    }}
\end{pchstack}
where $\mathcal{P}$ represent the space of valid \textit{parties} and $O(P, x_1, x_2, x_3, x_4)$ represents the AV calculation that will undergo each party and $P$ is the identity of the intended receiver.

We define the intermediate games $KEM_0$ and $KEM_1$, where intuitively each game represents an additional layer where the attacker behaves like the AM model: in game $KEM_0$, the attacker does not modify the flow of the last exchange and in game $KEM_1$ the attacker only modifies the first exchange.

\begin{pchstack}[ boxed , center, space=1em]
  {\procedure[linenumbering]{Game KEM$_0$}{\phantomsection\label{G00_KEM}
      Q \sample \mathcal{P} \\
      (\sk, \pk) \sample \texttt{KeyGen()} \\
      N \sample \{0,1\}^n \\
      (c, d) \sample \texttt{Com(N)}  \\
      (Q^*, c^*, d^*, \pk^*) \sample \adv_1^{O}(c, \pk, Q) \\
      ct \sample \texttt{Encaps($\pk^*$)} \\
      ct^* \sample \adv_2^{O}(c, c^*, d^*, \pk, \pk^*, Q, Q^*, ct) \\
      v_1 = O(Q, c, \pk, N, ct^*) \\
      N^{**} = Open(c^*, d) \\
      v_2 = O(Q^*, c^*, \pk^*, N^{**}, ct) \\
      b_1 = (v_1 == v_2) \\
      b_2 = \bigwedge (v_i \neq \bot) \\
      b_3 = (c^*, Q^*, \pk^*, ct^*) \neq (c, Q, \pk, ct) \\
      \pcreturn b_1 \wedge b_2 \wedge b_3
    }}
    
  {\procedure[linenumbering]{Game KEM$_1$}{\phantomsection\label{G10_KEM}
      Q \sample \mathcal{P} \\
      (\sk, \pk) \sample \texttt{KeyGen()} \\
      N \sample \{0,1\}^n \\
      (c, d) \sample \texttt{Com(N)}  \\
      (Q^*, c^*, d^*, \pk^*) \sample \adv_1^{O}(c, \pk, Q) \\
      ct \sample \texttt{Encaps($\pk^*$)} \\
      v_1 = O(Q, c, \pk, N, ct) \\
      N^{**} = Open(c^*, d) \\
      v_2 = O(Q^*, c^*, \pk^*, N^{**}, ct) \\
      b_1 = (v_1 == v_2) \\
      b_2 = \bigwedge (v_i \neq \bot) \\
      b_3 = (c^*, Q^*, \pk^*) \neq (c, Q, \pk^*) \\
      \pcreturn b_1 \wedge b_2 \wedge b_3
    }}
\end{pchstack}

Therefore, the value $\abs{\prob{\texttt{KEM$_{AV}$} \Rightarrow 1} - \prob{\texttt{KEM$_0$} \Rightarrow 1}}$ represents the advantage of the attacker in generating a collision between the oracle output from both parties when being able to modify the opening value received by the responder.

\textbf{Step 1) Analysis of $\abs{\prob{\texttt{KEM$_{AV}$} \Rightarrow 1} - \prob{\texttt{KEM$_0$} \Rightarrow 1}}$:} In order to bound this difference, let us define $\mathcal{F}$ as the event that, in Game KEM$_{AV}$, an attacker has the ability to exploit the additional step $\adv_3$ to win the game. Formally, this means the event the attacker is able to generate a value $d^{**}$ in step $\adv_3$ such that $Open(c^*, d^{**}) = N^{**}$, satisfying $v_1 = v_2 \wedge (c^*, Q^*, \pk^*, ct^*, d^{**}) \neq (c, Q, \pk, ct, d)$. 

Note that, by definition of the games under consideration, since the event $\mathcal{F}$ is the only difference between the games, the application of the difference lemma \cite{cryptoeprint:2004/332} yields
$$\abs{\prob{\texttt{KEM$_{AV}$} \Rightarrow 1} - \prob{\texttt{KEM$_0$} \Rightarrow 1}} \leq \prob{\mathcal{F} \Rightarrow 1}$$
Now, it is required to bound the probability of the above event. To do so, two events are considered: \texttt{F$_1$}, in which $(c^*, Q^*, \pk^*, ct^*) \neq (c, Q, \pk, ct)$, and \texttt{F$_2$}, in which $(c^*, Q^*, \pk^*, ct^*) = (c, Q, \pk, ct)$. \vspace{1mm}

\textbf{Event \texttt{F$_1$}:} Under \texttt{F$_1$}, $(Q, c, \pk, N, ct^*) \neq (Q^*, c^*, \pk^*, N^{**}, ct)$ and therefore the only possible way to achieve $v_1 = v_2$ is via an oracle collision. Therefore, an adversary needs to find $d^{**}$ such that $Open(c^*, d^{**}) = N^{**}$ and $O(Q, c, \pk, N, ct^*) = O(Q^*, c^*, \pk^*, N^{**}, ct)$.

For this analysis, two sub-cases are defined: given $N^* = Open(c^*, d^*)$, define \texttt{F$_{1, 1}$} as the event that $O(Q^*, c^*, \pk^*, N^*, ct) = O(Q, c, \pk, N, ct^*)$, and \texttt{F$_{1, 2}$} as the event that $O(Q^*, c^*, \pk^*, N^*, ct) \neq O(Q, c, \pk, N, ct^*)$. \vspace{1mm}

\textbf{Sub-event \texttt{F$_{1, 1}$}:} If event \texttt{F$_{1, 1}$} happens, the adversary has already won the game, as a random collision has happened. Therefore, $\prob{\texttt{F} \cond \texttt{F$_{1, 1}$}, \texttt{F$_{1}$}} = 1$ and $\prob{\texttt{F$_{1, 1}$} \cond \texttt{F$_{1}$}} = \frac{1}{2^{n_{av}}}$. \vspace{1mm}

\textbf{Sub-event \texttt{F$_{1, 2}$}:} If \texttt{F$_{1, 2}$} holds, the adversary is required to find $d^{**}$ such that $Open(c^*, d^{**}) = N^{**} \neq N^*$ satisfying $O(Q^*, c^*, \pk^*, N^{**}, ct) = O(Q, c, \pk, N, ct^*)$. Note that the requirement $N^{**} \neq N^*$ is precisely because $\texttt{F$_{1, 2}$}$ ensures that, if $N^{**} = N^*$, then $O(Q^*, c^*, \pk^*, N^{**}, ct) = O(Q^*, c^*, \pk^*, N^*, ct) \neq O(Q, c, \pk, N, ct^*)$. 

We claim that the probability $\prob{\texttt{F} \cond \texttt{F$_{1, 2}$}, \texttt{F$_{1}$}}$ is bounded by the \textit{Combined Collision} advantage: first note that the probability of an adversary to find a $d^{**}$ such that $Open(c^*, d^{**}) = N^{**} \neq N^*$ is modeled exactly by the Game G$_{Binding^*}$ defined above.

Now, we define an adversary $\mathcal{Y}_1$ that returns $N^{**}$ if it can find $d^{**}$ such that $Open(c^*, d^{**}) = N^{**} \neq N^*$, which by the above analysis, has probability of success bounded by $\advantage{\texttt{Binding}}{\texttt{CS}}$. Then, $\prob{\texttt{F} \cond \texttt{F$_{1, 2}$}, \texttt{F$_{1}$}}$ is exactly $\advantage{\texttt{combined}}{\texttt{CHF}}[(\mathcal{A}, n_{av}, \mathcal{Y}_1)]$ and thus by \ref{combined_adv} $$\prob{\texttt{F} \cond \texttt{F$_{1, 2}$}, \texttt{F$_{1}$}} \leq q_O \frac{\advantage{\texttt{Binding}}{\texttt{CS}}}{2^{n_{av}}}$$

\textbf{Event \texttt{F$_2$}:} Under \texttt{F$_2$}, $(c^*, Q^*, \pk^*, ct^*) = (c, Q, \pk, ct)$ and therefore there are only two possible ways to win the game \texttt{F}. The first one, referred to as \texttt{F$_{2, 1}$}, is to generate a collision on the inputs to the oracle function. This means, to find $d^{**} \neq d$ such that $Open(c^*, d^{**}) := N^{**} = N =: Open(c, d)$. Note that the requirement $d^{**} \neq d$ is to ensure that a collision on the modified elements does not occur. \vspace{1mm} 

\textbf{Sub-event \texttt{F$_{2, 1}$}:} The probability of such a collision is modeled exactly by the game G$_{Strong-Binding^*}$ defined above. Therefore
$$\prob{\texttt{F$_{2, 1}$} \cond \texttt{F$_2$}} = \prob{\texttt{G$_{Strong-Binding^*}$} \Rightarrow 1} \leq \advantage{\texttt{Strong-Binding}}{\texttt{CS}}$$

The second one, referred to as \texttt{F$_{2, 2}$}, is to generate a collision on the outputs, provided that the inputs are not equal. This means, to find $d^{**} \neq d$ such that $Open(c^*, d^{**}) = N^{**} \neq N$ and $O(Q^*, c^*, \pk^*, N^{**}, ct) = O(Q, c, \pk, N, ct^*)$. \vspace{1mm}

\textbf{Sub-event \texttt{F$_{2, 2}$}:} In a similar manner to above, we claim that the probability $\prob{ \texttt{F$_{2, 2}$} \cond \texttt{F$_{2}$}}$ is bounded by the \textit{Combined Collision} advantage: first note that the probability of an adversary to find a $d^{**} \neq d$ such that $Open(c, d^{**}) = N^{**} \neq N$ is modeled exactly by the Game G$_{Binding^*}$.

Now, we define an adversary $\mathcal{Y}_2$ that returns $N^{**}$ if it can find $d^{**} \neq d$ with the above conditions. Then, $\prob{\texttt{F$_{2, 2}$} \cond \texttt{F$_{2}$}}$ is exactly $\advantage{\texttt{combined}}{\texttt{CHF}}[(\mathcal{A}, n_{av}, \mathcal{Y}_2)]$ and thus by \ref{combined_adv} 
$$\prob{ \texttt{F$_{2, 2}$} \cond \texttt{F$_{2}$}} \leq \prob{\texttt{G$_{Binding^*}$} \Rightarrow 1} \cdot \frac{q_O}{2^{n_{av}}} \leq \advantage{\texttt{Binding}}{\texttt{CS}} \cdot \frac{q_O}{2^{n_{av}}}$$

Together, we have 
\begin{align*}
& \prob{\texttt{F} \Rightarrow 1} = \prob{\texttt{F} \Rightarrow 1 \cond \texttt{F$_1$}} \cdot \prob{\texttt{F$_1$}} + \prob{\texttt{F} \Rightarrow 1 \cond \texttt{F$_2$}} \cdot \prob{\texttt{F$_2$}} \leq \\
& \prob{\texttt{F} \Rightarrow 1 \cond \texttt{F$_{1, 1}$}, \texttt{F$_1$}} \cdot \prob{\texttt{F$_{1, 1}$} \cond \texttt{F$_1$}} + \prob{\texttt{F} \Rightarrow 1 \cond \texttt{F$_{1, 2}$}, \texttt{F$_1$}} \cdot \prob{\texttt{F$_{1, 2}$} \cond \texttt{F$_1$}} + \prob{\texttt{F} \Rightarrow 1 \cond \texttt{F$_2$}} \leq \\
& \frac{1}{2^{n_{av}}}\cdot (1 + q_O \cdot \advantage{\texttt{Binding}}{\texttt{CS}}) + \prob{\texttt{F$_{2, 1}$} \Rightarrow 1 \cond \texttt{F$_2$}} + \prob{\texttt{F$_{2, 2}$} \Rightarrow 1 \cond \texttt{F$_2$}} \leq \\
& \frac{1}{2^{n_{av}}}\cdot (1 + 2 \cdot q_O \cdot \advantage{\texttt{Binding}}{\texttt{CS}}) +\advantage{\texttt{Strong-Binding}}{\texttt{CS}}
\end{align*}

The value $\abs{\prob{\texttt{KEM$_0$} \Rightarrow 1} - \prob{\texttt{KEM$_1$} \Rightarrow 1}}$ represents the advantage of the attacker in generating a collision between the oracle output from both parties when being able to modify the random challenge value sent by the responder. \vspace{1mm}

\textbf{Step 2) Analysis of $\abs{\prob{\texttt{KEM$_{0}$} \Rightarrow 1} - \prob{\texttt{KEM$_1$} \Rightarrow 1}}$:} In order to bound this difference, let us define $\mathcal{H}$ as the event that in game \texttt{KEM$_0$} an adversary has the ability to exploit the second exchange to generate $ct^*$ in step $\adv_2$ such that $v_1 = v_2$ and $(c^*, Q^*, \pk^*, ct^*) \neq (c, Q, \pk, ct)$. Note that, by definition of the games under consideration, the application of the difference lemma \cite{cryptoeprint:2004/332} yields
$$\abs{\prob{\texttt{KEM$_{0}$} \Rightarrow 1} - \prob{\texttt{KEM$_1$} \Rightarrow 1}} \leq \prob{\mathcal{H} \Rightarrow 1}$$.
Now, it is required to bound the probability of the above event. There are exactly two ways: via a random generation of $ct^*$ and collision between outputs (referred to as \texttt{H$_1$}), and by trying to extract information and generating a targeted value for $ct^*$ (i.e. not trying random collision), referred to as \texttt{H$_2$}.

\textbf{Event \texttt{H$_1$}:} The advantage in winning the game from the first option is bounded by the probability of a random collision, which is $\frac{1}{2^{n_{av}}}$. Therefore, we will focus on the probabilities for an adversary in making an elaborated guess. \vspace{1mm}

\textbf{Event \texttt{H$_2$}:} The analysis is done in terms of the event that the adversary knows the value of $N$, referred to as $D$. First, we will analyze the probability of $D$ itself. Since $N = Open(c, d)$, the pair $(c, d)$ represents a valid opening of $N$. Therefore, this probability is bounded by the advantage in retrieving the committed input from a valid commitment. This is modeled by the game G$_{Hiding^*}$ above. Therefore, via the same argument as in the proof of \ref{3_pass_2_pass_emulation_bound}, it holds that $\prob{D} \leq \advantage{\texttt{Hiding}}{\texttt{CS}}$. \vspace{1mm}

\textbf{Sub-event \texttt{$\neg D$}:} In case $D$ does not hold, the adversary does not possess all the information required to make an informed guess: the value $v_1$ is dependent on both the choice of $ct^*$ of the adversary and $N$, which is defined but not known to the adversary. Therefore, $\prob{\texttt{H$_2$} \cond \neg D} = 0$, as the only possible way when not all information is known is random collision. \vspace{1mm}


\textbf{Sub-event \texttt{$D$}:} In case $D$ does hold, we analyze the success probability in terms of the event $C$ that $c^* = c$. If $C$ is also true, then $Open(c^*, d) = N^{**} = N = Open(c, d)$. Therefore, since $N$ is known, $N^{**}$ is also known. This provides the adversary with the knowledge of all elements involved in the generation of both $v_1$ and $v_2$. 

Thus the success probability $\prob{\texttt{H$_2$} \cond C,D}$ resides in the adversaries ability to generate a random element $ct^*$ such that 
$$O(Q^*, c, \pk^*, N, ct) = O(Q, c, \pk, N, ct^*)$$ 
Since the attacker cannot win the game via input collision, as it must hold that $(Q^*, c^*, \pk^*, ct^*) \neq (Q, c, \pk, ct)$, the only viable option is to search for a preimage-collision over valid encapsulation values $ct^*$, which has probability at most $\frac{q_O}{2^{n_{av}}}$, with $q_O$ being the maximum number of queries made to the oracle $O$.

In case $D$ holds but $C$ does not, further information cannot be claimed about $N^{**}$. For any adversary to be able to win the game without random collision, three conditions must happen simultaneously: first, $N^{**} \neq \bot$. Second, to be able to retrieve $N^{**}$ from $c^*$. Lastly, the ability to force a collision via generating $ct^*$.

For the analysis of first condition, since $N^{**} = Open(c^*, d)$ and $c^* \neq c$, it is required to bound the probability that $N^{**} \neq \bot$. It can be bounded by the advantage of the CR property of commitment schemes. Note, however, that this bound is very conservative, since the CR-CS game allows the adversary to find $(c, c', d)$, while in here they are all established.

The second condition requires to bound the probability of retrieving its value. Since $N^{**} = Open(c^*, d) \neq \bot$, the pair $(c^*, d)$ conforms a valid commitment pair of the value $N^{**}$. Therefore, this probability is bounded by the advantage in retrieving the committed input from a valid commitment. As above, this can be bounded by the Hiding property of the CS.

The third condition requires to bound the probability of being able search for a preimage-collision over encapsulation values $ct^*$. It has probability at most $\frac{q_O}{2^{n_{av}}}$, with $q_O$ being the maximum number of queries made to the oracle $O$.

Together, this means that $$\prob{\texttt{H$_2$} \cond \neg C,D} \leq \advantage{\texttt{Hiding}}{\texttt{CS}} \cdot \advantage{\texttt{CR}}{\texttt{CS}} \cdot \frac{q_O}{2^{n_{av}}}$$
and therefore
\begin{align*}
& \prob{\texttt{H} \Rightarrow 1} \leq \prob{\texttt{H$_1$}} + \prob{\texttt{H$_2$}} \leq \\
& \frac{1}{2^{n_{av}}} + \prob{\texttt{H$_2$} \cond \neg D} \cdot \prob{\neg D} + \prob{\texttt{H$_2$} \cond D} \cdot \prob{D} = \\
& \frac{1}{2^{n_{av}}} + (\prob{\texttt{H$_2$} \cond C, D} \cdot \prob{C \cond D} + \prob{\texttt{H$_2$} \cond \neg C, D} \cdot \prob{\neg C \cond D}) \cdot \prob{D} \leq \\
& \frac{1}{2^{n_{av}}} + (1 + \advantage{\texttt{Hiding}}{\texttt{CS}} \cdot \advantage{\texttt{CR}}{\texttt{CS}}) \cdot \frac{q_O}{2^{n_{av}}} \cdot \advantage{\texttt{Hiding}}{\texttt{CS}}
\end{align*}

\textbf{Step 3) Analysis of $\prob{\texttt{KEM$_1$} \Rightarrow 1}$:} Lastly, the advantage of the game $KEM_{1}$ is exactly the probability of a random collision between inputs when generating a different commitment value $c^*$, a different public key value $\pk^*$ or when specifying a different receiver $Q^*$ than the one $Q$ originally intended, which is $\frac{1}{2^{n_{av}}}$.

Together, we have the following:
\begin{align*}
    & \abs{\prob{\texttt{KEM$_{AV}$} \Rightarrow 1}} \leq \\ 
    & \abs{\prob{\texttt{KEM$_{AV}$} \Rightarrow 1} - \prob{\texttt{KEM$_0$} \Rightarrow 1}} + \abs{\prob{\texttt{KEM$_{0}$} \Rightarrow 1} - \prob{\texttt{KEM$_1$} \Rightarrow 1}} + \abs{\prob{\texttt{KEM$_1$} \Rightarrow 1}} \leq \\
    & \frac{3}{2^{n_{av}}} + \frac{q_O}{2^{n_{av}}} \cdot ((1 + \advantage{\texttt{Hiding}}{\texttt{CS}} \cdot \advantage{\texttt{CR}}{\texttt{CS}}) \cdot \advantage{\texttt{Hiding}}{\texttt{CS}} \\
    & + 2 \cdot \advantage{\texttt{Binding}}{\texttt{CS}} +\advantage{\texttt{Strong-Binding}}{\texttt{CS}})
\end{align*}
    Then, considering all possible messages between parties, which equals the maximum number of messages sent by each party times the square of the number of parties, we arrive to the desired bound.
\end{proof}

It is important to note that this protocol only provides one-way authentication, as only one AV calculation is employed.

Now, the SK security of the 3-pass KEM-based protocol is stated as follows:
\begin{proposition}\label{gen_3_pass_KEM_sk_security}
The 3-pass KEM-based protocol is $(\epsilon + 2\cdot \alpha)$-SK secure in the UM,
    where $\alpha$ is the emulation bound proven in Proposition \ref{3_pass_2_pass_emulation_bound_kem} and $\epsilon$ is the SK-security bound proven in \ref{sk_security_2_pass_kem}.
\end{proposition}
\begin{proof}
    The result is a direct application of \ref{ThEmulSK} and \ref{3_pass_2_pass_emulation_bound_kem}.
\end{proof}

\section{Conclusions}
In this work we have shown that 3-pass protocols secure under the commitment-based AKE model of \cite{RMSL:2026:AKE_COM} exist for both KA-based and KEM-based primitives, settling the only open case between the 2-pass impossibility result and the 4-pass constructions of \cite{RMSL:2026:AKE_COM}. The protocols are constructed ad hoc, outside the reach of the commitment-based MT compiler, but follow the same design principles and admit security proofs of the same form, grounded in the hiding, binding, strong-binding, and collision-resistance properties of the underlying commitment scheme and the random oracle model.

The central design insight is that a single commitment, placed at the outset of the exchange before the responder has contributed any information, is sufficient to bind the initiator to its intended values and prevent adversarial manipulation of the AV verification. In the KA-based case this commitment is placed over the initiator's public key directly. In the KEM-based case, where committing to the public key would block the responder from encapsulating without an additional pass, the commitment is instead placed over a fresh random nonce included in the AV computation. The two constructions illustrate how the structural differences between KA and KEM primitives, already highlighted in \cite{RMSL:2026:AKE_COM}, continue to surface at the level of protocol design even when the security arguments run in parallel.

The trade-off incurred by the reduction to 3 passes is one-way rather than mutual authentication: the responder authenticates the initiator, but the initiator has no symmetric guarantee. This is an inherent consequence of the asymmetry of a 3-pass exchange and not a weakness of the specific constructions. Whether a 3-pass protocol achieving mutual authentication is possible under this model, or whether mutual authentication requires at least 4 passes, remains an open question.

More broadly, this work confirms that the commitment-based AKE model of \cite{RMSL:2026:AKE_COM} is flexible enough to admit efficient ad hoc constructions beyond those derivable from its compiler, and that the game-based proof techniques developed there extend naturally to this setting. Identifying further protocol families that fall within the model, and understanding the precise trade-offs between the number of passes, the authentication guarantees achievable, and the underlying primitive, are natural directions for future work.

\bibliography{lib.bib}
\bibliographystyle{plain}

\nocite{AGKS:2005:TKD}
\nocite{BCK:1998:MAD}
\nocite{CK:2001:AKE}
\nocite{CS:2003:DAP}
\nocite{Dent:2003:DGK}

\end{document}